\documentclass[11pt,a4paper]{article}
\usepackage{dsfont}
\usepackage{mathtools}
\usepackage{enumitem}
\usepackage{hyperref, comment,url}
\usepackage[ruled,vlined,linesnumbered]{algorithm2e}
\usepackage{amsfonts}
\usepackage{tikz}
\usepackage{graphicx}
\usepackage{float}

\usepackage[margin=1in]{geometry}
\usepackage{amsmath}
\usepackage{amssymb}
\usepackage{amsthm}
\usepackage{microtype}
\usepackage[latin1,utf8]{inputenc}
\usepackage{authblk}

\hypersetup{
    colorlinks,
    linkcolor={red!50!black},
    citecolor={blue!50!black},
    urlcolor={blue!80!black}
}
\newtheorem{theorem}{Theorem}[section]

\newtheorem{lemma}[theorem]{Lemma}
\newtheorem{corollary}[theorem]{Corollary}
\newtheorem{claim}[theorem]{Claim}

\newtheorem{definition}{Definition}
\newtheorem{remark}{Remark}

\newcommand\N{\mathbb{N}}
\newcommand\E{\mathop{\mathbb{E}}}
\newcommand\R{\mathbb{R}}

\newcommand\bfx{\textbf{x}}
\newcommand\bfy{\textbf{y}}

\newcommand\bfw{\textbf{w}}

\newcommand\mcQ{\mathcal{Q}}

\newcommand\mcG{\mathcal{G}}
\newcommand\mcC{\mathcal{C}}
\newcommand\mcE{\mathcal{E}}
\newcommand\mcA{\mathcal{A}}
\newcommand\mcX{\mathcal{X}}

\newcommand\Del[2]{\textup{Del}_{{#1}, {#2}}}
\newcommand\Insert[2]{\textup{Ins}_{{#1}, {#2}}}
\newcommand{\Bernoulli}{\textup{Bern}}
\newcommand\sgn[1]{\textup{sign}\left({#1}\right)}

\newcommand{\NumStates}{S}
\newcommand{\NumSamples}{N}

\newcommand\wt[1]{\widetilde{#1}}
\newcommand\abs[1]{{\left\lvert{#1}\right\rvert}}

\newcommand\Prob[2]{{\Pr_{#1}\left[ {#2} \right]}}

\newcommand\EE[1]{{\mathop{\mathbb{E}}\left[ {#1} \right]}}
\newcommand\Expect[2]{{\mathop{\mathbb{E}}_{#1}\left[ {#2} \right]}}

\DeclareMathOperator{\Var}{Var}

\newcommand\supp{\textup{Supp}}
\newcommand{\polylog}{\textnormal{polylog}}
\newcommand\poly{\textnormal{poly}}
\newcommand\defeq{\stackrel{\textup{def}}{=}}

\newcommand\set[1]{{\left\{ #1 \right\}}}
\newcommand{\bits}{\set{0, 1}}

\newcommand{\eg}{{e.g.}}

\newcommand{\ie}{{i.e.}}

\newcommand{\iid}{{i.i.d.}}
\newcommand{\etal}{et al.}

\usepackage{soul}

\newcommand{\QP}{\textup{QP}}
\newcommand{\Norm}[1]{\left\lVert {#1}\right\rVert}
\newcommand{\Geom}{\textup{Geom}}
\newcommand{\Ber}{\textup{Ber}}

\newcommand{\IP}[2]{\left\langle {#1} , {#2} \right\rangle}

\newcommand{\Paren}[1]{\left(#1\right)}

\newcommand{\DelChannel}{\textup{Deletion Channel}}
\newcommand{\InsChannel}{\textup{Insertion Channel}}

\usepackage{fullpage}

\title{The Quasi-probability Method and Applications for Trace Reconstruction}

\author[]{Ittai Rubinstein}

\affil[]{MIT EECS and CSAIL, Cambridge MA}
\date{\today}

\begin{document}

\maketitle

\begin{abstract}
    In the {\em trace reconstruction problem}, one attempts to reconstruct a fixed but unknown string $\bfx$ of length $n$ from a given number of {\em traces} $\wt{\bfx}$ drawn iid from the application of a noisy process (such as the deletion channel) to $\bfx$.
The best known algorithm for the trace reconstruction from the deletion channel is due to Chase, and recovers the input string whp given $\exp\left( \wt{O}\left(n^{1/5}\right) \right)$ traces~\cite{chase2021separating}.

The main component in Chase's algorithm is a procedure for {\em $k$-mer estimation}, which, for any marker $\bfw \in \bits^k$ of length $k$, computes a ``smoothed'' distribution of its appearances in the input string $\bfx$~\cite{cheng2023k,mazooji2024substring}.
Current $k$-mer estimation algorithms fail when the deletion probability is above $1/2$, requiring a more complex analysis for Chase's algorithm.
Moreover, the only known extension of these approaches beyond the deletion channels is based on numerically estimating high-order differentials of a multi-variate polynomial, making it highly impractical~\cite{rubinstein2022average}.

In this paper, we construct a simple Monte Carlo method for $k$-mer estimation which can be easily applied to a much wider variety of channels.
In particular, we solve $k$-mer estimation for any combination of insertion, deletion and bit-flip channels, even in the high deletion probability regime, allowing us to directly apply Chase's algorithm for this wider class of channels.

To accomplish this, we utilize an approach from the field of quantum error mitigation (the process of using many measurements from noisy quantum computers to simulate a clean quantum computer), called the {\em quasi-probability method} (also known as probabilistic error cancellation)~\cite{temme2017error,piveteau2022quasiprobability}.
We derive a completely classical version of this technique, and use it to construct a $k$-mer estimation algorithm.

Though our algorithm is ``quantum-inspired'', no background in quantum computing is needed to understand this paper.

\end{abstract}

\newpage

\thispagestyle{empty}

\section{Introduction}\label{sec:introduction}

\subsection{The Trace Reconstruction Problem}

Trace reconstruction is a fundamental problem in computational biology, which attempts to model the challenges in aligning multiple DNA sequences to a common ancestor. In this problem, a binary string $\bfx$ of length $n$, is subjected to a random synchronisation channel such as a {\em deletion channel}, an {\em insertion channel} or a {\em symmetry channel}. In a deletion channel, each bit of the input string is independently deleted with probability $\delta$, and in an insertion channel $\text{Geom}(1 - \eta)$ independent uniformly distributed bits are inserted before each bit of the message, and the symmetry channel with parameter $\sigma$ flips the value of each input bit independently with probability $\sigma$. 
We call the output of such a synchronisation channel a {\em trace}.

The trace reconstruction problem attempts to solve the following problem: given a set of traces $\wt{\bfx}^1, \ldots, \wt{\bfx}^\NumSamples \sim \mcC(\bfx)$, reconstruct the input string $\bfx$.

This problem has been extensively studied since the early 2000s \cite{batu2004reconstructing}. 
There are two main versions of the problem: the {\em worst-case} and the {\em average-case}. In the worst-case, the input string $\bfx$ is adversarially chosen, and the reconstruction algorithm must work for all possible strings of length $n$. In the average-case, $\bfx$ is chosen uniformly at random from all possible strings of length $n$.

McGregor~\etal~\cite{mcgregor2014trace} and Rubinstein~\cite{rubinstein2022average} showed that the average and worst-case versions of the problem are essentially equivalent (McGregor showed that if $h(n)$ traces are needed for worst-case trace reconstruction, then at least $h(\log(n))$ traces are needed for average-case trace reconstruction and Rubinstein showed that if $h(n)$ traces suffice for a problem similar to worst-case trace reconstruction, then $h(O(\log(n)))$ suffice for the average-case).
For the rest of this paper we will focus on the worst-case trace reconstruction.

Chase~\cite{chase2021new} proved that in the worst-case at least $n^{3/2} / \poly(\log n)$ traces are required to reconstruct a string of length $n$, improving upon the previous bound of $n^{5/4} / \poly(\log n)$ by Holden and Lyons \cite{holden2020lower}. Holenstein et al. \cite{holenstein2008trace} established an upper bound of $\exp(\widetilde{O}(n^{1/2}))$ on the sample complexity of the worst-case trace reconstruction problem. This was improved to $\exp(O(n^{1/3}))$ by Nazarov and Peres \cite{nazarov2017trace}, De et al. \cite{de2017optimal}, and later by Chase \cite{chase2021separating}, who improved the bound to $\exp(\widetilde{O}(n^{1/5}))$.

The ability to easily adapt a trace reconstruction algorithm to more complex error models is crucial for any real-world implementation, as it is highly unlikely that we will be faced with a pure form of the problem in practice.
While Nazarov and Peres and De et al.'s algorithms are directly applicable to the insertion-deletion-symmetry channel and were extended to an even wider class of channels by Cheraghchi et al.~\cite{cheraghchi2022mean}, Chase's analysis is limited to the deletion channel.

The key step in Chase's trace reconstruction algorithm is a process later described as the ``$k$-mer estimation'' problem~\cite{mazooji2024substring,cheng2023k}.
The $k$-mer statistic of a string $\bfx$ is roughly defined as a mapping from any marker $\bfw \in \bits^k$ of length $k$ to a ``smoothed'' set of the indices where $\bfw$ appears in the input string $\bfx$.

Current $k$-mer estimation approaches fail in the high deletion probability regime $\delta \geq 1/2$.
Chase overcomes this limitation in his algorithm, but this requires an extension of the complex analysis involved, and an increased time complexity for this regime.

Rubinstein~\cite{rubinstein2022average} extended Chase's analysis to insertion-deletion-symmetry channels by generalising the complex analysis approach to $k$-mer estimation, but the analysis is highly technical and the resulting algorithm, which is based on numerically estimating a high-order derivative of a multivariate polynomial to a very high degree of accuracy, would be challenging to implement in practice.

In this paper, we use the quasi-probability method to construct a much simpler and more robust version of the $k$-mer estimation algorithm.
The quasi-probability based $k$-mer estimation algorithm can be easily applied to a wide variety of channels, including insertion, deletion and bit-flip channels.

Moreover, we will show that the quasi-probability method is inherently self-composeable.
Informally, this means for any two noise channels $\mcC_1, \mcC_2$ and quasi-probability ``solutions'' $\mcQ_1, \mcQ_2$ for each of these channels, we can compose them to construct a quasi-probability method $\mcQ$ for the composed noise channel $\mcC = \mcC_1 \circ \mcC_2$.

This will allow us to extend this simple $k$-mer estimation to any combination of insertion, deletion and bit-flip channels, even in the high-deletion regime $1/2 \leq q < 1$.

\begin{theorem}[Main Result (informal)]
\label{thm:main_res_inf}
    Let $\mcC_1, \ldots, \mcC_L$ be any set of insertion / deletion / symmetry channels.
    Then there exists an algorithm $\mcA$ with runtime 
    and sample complexity $\exp \left(\wt{O}\left( n^{1/5} \right)\right)$ for trace reconstruction from the composite channel $\mcC = \mcC_1 \circ \cdots \circ \mcC_L$.
\end{theorem}

\begin{remark}
    Throughout the paper, we suppress dependencies on the channel parameters in big-O notations (i.e., our dependence on the number of insertion-deletion channels and their parameters). This is common practice in trace reconstruction literature.
\end{remark}

Finally, note that in addition to its simplicity and generalizability, our improved $k$-mer estimation algorithm also allows for a reduced time complexity in Chase's trace reconstruction algorithm.
In his original paper, Chase does not analyze the runtime of his reconstruction algorithm, but focuses on bounding the sample complexity.
Rubinstein shows that a straightforward implementation of Chase's algorithm would have a time complexity of $\exp \left(\wt{\Theta}\left( n^{4/5} \right)\right)$ when $q < 1/2$, and a complexity of $\exp \left(\Theta\left( n \right)\right)$ when applying Chase's solution for higher deletion probability.
Our quasi-probability based algorithm obtains a lower time complexity even for high deletion probabilities and more general combinations of channels.

\paragraph{Trace Reconstruction and Population Recovery}

Another interesting version of the trace reconstruction is trace reconstruction from a population~\cite{ban2019beyond,narayanan2021improved}.
In this version of the trace reconstruction problem, instead of having a fixed input string $\bfx$, each trace draws an input string $\bfx$ independently from some unknown distribution $\mcX$ over strings of length $n$, and applies the noise channel $\mcC$ to it.
The output of the reconstruction algorithm also changes, and our goal becomes to compute a distribution that is close to $\mcX$ (say in TVD).

The difficulty of population recovery from a deletion channel is often parameterized in terms of both the length of the input strings, and the size of the support of $\mcX$.
Ban et al.~\cite{ban2019beyond} show that for $\ell = \abs{\supp \set{\mcX}} \leq n^{0.499}$, at least $n^{\Omega(\ell)}$ traces are needed to reconstruct $\mcX$ to a reasonable degree of accuracy and Narayanan~\cite{narayanan2021improved} provides an algorithm for population recovery from $\exp\left(O(n^{1/3} \ell^2) \right)$.

The main technical component in our $k$-mer estimation algorithm will solve a problem slightly harder than population recovery for populations with large support.

\begin{theorem}[Recovery of Large Populations from Synchronisation Channels]
\label{thm:pop_recovery}
    Let $\mcC_1, \ldots, \mcC_L$ be any set of insertion, deletion, symmetry channels, and let $\mcC = \mcC_1 \circ \cdots \circ \mcC_L$ be their composition.

    There exists an algorithm $\mcA$ that, with high probability, recovers any population $\mcX$ of binary strings of length $n$ to within a TVD of $\varepsilon$, given $\exp\left(O(n)\right) \poly(\varepsilon^{-1})$ samples of $\mcC(\mcX)$, regardless of the size of the support of $\mcX$.
\end{theorem}

The main focus of our paper will be to prove Theorem~\ref{thm:main_res_inf} by drawing inspiration from the field of quantum error mitigation.

\subsection{Quantum Error Mitigation and the Quasi-probability Method}

Quantum computing promises to provide exponential speed-ups on several crucial computational problems, such as integer factoring and simulation of quantum systems.
However, as we are now entering the Noisy Intermediate-Scale Quantum (NISQ) era, current quantum computers suffer from errors and are too small for error correction~\cite{Preskill2018quantumcomputingin}.

One approach to utilizing these NISQ devices is ``quantum error mitigation’’, loosely defined as any method of using many samples of noisy quantum processes to reconstruct some aspect of the behavior of a noiseless process~\cite{cai2023quantum}.
Many approaches to quantum error mitigation have been studied theoretically and are used in practice in experiments on quantum devices.
Here, we will focus on a specific approach called the ``quasi-probability method'' (also know as ``probabilistic error cancellation’’  or PEC) \cite{temme2017error,piveteau2022quasiprobability}.

Our main observation will be that the quasi-probability method is not limited to the quantum setting and can be applied to classical channels on a classical computer, allowing for several potential applications in learning theory problems.

In particular, we show that this approach can be used to simplify leading trace reconstruction algorithms and extend them to more complex noise channels.

Quasi-probability sampling or probabilistic error cancellation provides a framework for adding additional errors to a noisy process in such a way that when averaging over the samples, these additional errors will cancel the bias caused by the original errors (hence the name ``probabilistic error cancellation'').
This can seem somewhat counter-intuitive, so before diving into a detailed derivation of the general method in Section~\ref{sec:PEC}, we devote this subsection to giving an overview of the behaviour of the quasi-probability method in 2 simple settings.

\paragraph{Example 1: Single-Fault Experimental System}

Suppose we are given access to some experimental apparatus which outputs a real-valued sample $s \in [-1, 1]$ from some unknown distribution $S$.
After running this experiment $N$ times, we could gather our $N$ \iid~samples $s_1, \ldots, s_N$ and average them out to estimate $ s^* \defeq \sum_i s_i / N \approx \mu_s $.
An application of Hoeffding's inequality tells us that with high probability this estimate will be accurate up to a statistical error of order $O(1/\sqrt{N})$.

Now, suppose someone tells us that our experimental device is faulty and undergoes some error $E$ \iid~with some known probability $\varepsilon$ in each experiment.
We denote the output distribution of this faulty device by $S^\prime$.
If we simply average our samples, then $\varepsilon$ fraction of them would be faulty and we would incur a bias of order $\varepsilon$ to our estimate $s^*$.

For instance, consider a chemical engineer performing several iterations of a new reaction to estimate a specific property of the resulting solution, but whose lab assistant forgets to add the last ingredient with some known probability $\varepsilon$.
In this case, $E$ would represent the outcome of the experiment when the lab assistant forgets to include this ingredient, and $S^\prime$ would be the overall distribution of experiments run by this engineer.

If we are only allowed to sample from the original process and know nothing about how this error $E$ affects its output, then there isn't much more we can do.
The quasi-probability method requires only the additional ability to trigger this fault, causing the device to ``fail'' on purpose, in which case we assume that our sample comes directly from the error distribution $E$.

Neglecting terms of order $\varepsilon^2$ for the moment, the QP method described in Section~\ref{sec:PEC} will tell us to sample $s \leftarrow S^\prime$ from the faulty device  w.p.  $1-\varepsilon$.
On average, this will result in a single batch of approximately $(1-2\varepsilon) N$ ideal samples and $\varepsilon N$ with an unintentional error.
In the remaining $\varepsilon$ fraction of the cases, we trigger the error (causing us to sample from the error distribution $E$), but then flip the sign of the result.

In other words, each QP sample is selected \iid~at random from the following process:
\begin{enumerate}
    \item Select a bit $b \leftarrow \text{Bernoulli}(\varepsilon)$. This bit will decide whether we cause an intentional error (w.p. $\varepsilon$) or not (w.p. $1-\varepsilon$).
    \item Gather a sample from the chosen distribution $s \leftarrow \begin{cases} S^\prime & b = 0 \\ E & b = 1 \end{cases}$
    \item Output the sample with an appropriate sign $(-1)^{b} \cdot s$
\end{enumerate}

Averaging over these samples, we get
\[
\hat{s} = \frac{\sum_{1 \leq i \leq N} s_i (-1) ^{b_i}}{N} \approx (1-\varepsilon) \langle S^\prime \rangle - \varepsilon \langle E \rangle = (1-2\varepsilon) \langle S \rangle + \varepsilon \langle E \rangle - \varepsilon \langle E \rangle = (1-2\varepsilon) \langle S \rangle
\]

On average, the two sources of erroneous samples cancel out, while the $\approx (1-2\varepsilon) N$ error-free samples are still taken into account for the average.
Therefore, $\hat{s}$ is an unbiased estimator for $(1-2\varepsilon) \mu_s$.
In order to obtain an estimation of $\mu_s$, we simply rescale $\hat{s}$ by a factor of $\gamma = (1-2\varepsilon)^{-1}$.

This gives us a black-box approach to estimating the average of an error-free process given access to a single-fault version of the process and the ability to cause an error on purpose.

\paragraph{Example 2: Multiple Faults and Recursive Construction}

Perhaps the greatest strength of the QP method is that it is self-composable, allowing us to apply it recursively.
Expanding on the previous example, we could imagine that our experimental apparatus can undergo any number of errors from a given list $E_1, \ldots, E_k$, and that each of them can occur independently w.p. $\varepsilon_1, \ldots, \varepsilon_k$ (\eg~a bit-flip might occur at any point of a long computation).

For the sake of this example, we will neglect terms of order $\sum_i \varepsilon_i ^2 \approx k \varepsilon^2$, but not terms of order $\left(\sum_i \varepsilon_i\right)^2 \approx k^2 \varepsilon^2$ (where $\varepsilon$ is the order of magnitude of the individual error probabilities $\varepsilon_i$).
In particular we will neglect any bias due to the event of probability $\leq k\varepsilon^2$ that the same error $E_i$ occurs twice (once because of the intentional errors prescribed by the QP sampling and a second time due to the noise in the system).
Moreover, we assume that by intentionally causing one error $E_i$, we do not affect the probability that any other error $E_j$ will occur.

A standard QP sampling for this scenario would look like this:
\begin{enumerate}
    \item Select a tuple of independently distributed bits $\left(b_i \leftarrow \text{Bernoulli}(\varepsilon_i)\right)_{1 \leq i \leq k}$
    \item Draw a sample $s$ from the experimental apparatus, intentionally causing each error $E_i$ if $b_i = 1$.
    \item Compute the sign $(-1)^{b_1 \oplus \cdots \oplus b_k}$.
    \item Rescale the result by a factor of $\gamma = \prod_i \left(1 - 2 \varepsilon_i\right)^{-1}$.
    \item Output the signed and rescaled sample $s_{\text{QP}} = (-1)^{b_1 \oplus \cdots \oplus b_k} \gamma s $.
\end{enumerate}

In order to show that the output distribution is unbiased (up to an error of order $\sum_i \varepsilon_i ^2$), we consider any non-empty set of possible errors $E_{i_1}, \ldots, E_{i_\ell}$ that may occur (with some of the errors occurring organically through the noisy process and some due to the QP sampling above).
Because any combination of ``intentional'' and ``unintentional'' error that produce this specific pattern is equally likely, conditioned on this specific pattern of errors, the number of intentional errors $n_b = b_{i_1} + \cdots + b_{i_\ell} \leftarrow \text{Bin}(1/2, \ell)$ is binomially distributed.

From here, it is easy to show that the sign associated with this QP sample $(-1)^{n_b} \leftarrow (-1)^{\text{Bin}(1/2, \ell)}$ has expectation $0$.
Therefore for any non-empty set of errors $E_{i_1}, \ldots, E_{i_\ell}$, their contribution to the average of the QP samples is $0$.

Finally, because each step of the process adds an error w.p. $2\varepsilon_i$, the ideal process is sampled only with probability $\gamma^{-1} = \prod_i (1 - 2\varepsilon_i)$, requiring the rescaling step.

\paragraph{Sampling Overhead}
Naturally, removing noise from our samples must come at some cost.
In particular, much like other quantum error mitigation techniques, QP comes with a ``sampling overhead''.

Assume that the ideal and noisy distributions are all bounded (\eg~that all the samples from the device are within the segment $[-1, 1]$).
Then $\Theta\left(\Delta^{-2}\right)$ samples suffice to estimate the expectation $\mu_S = \langle S \rangle$ to within an error of $\pm\Delta$.

However, the rescaling step multiplies our output by a factor of $\gamma = \prod_i \left(1 - 2 \varepsilon_i\right)^{-1} \approx \exp(2k\varepsilon)$.
Therefore, the variance of the QP samples could be of order $\Theta(\gamma^2)$, requiring $\Theta\left(\gamma^2 \Delta^{-2}\right)$ samples in order to obtain the same $\pm \Delta$ approximation of the expectation $\mu_S$ of the ideal process.
This exponential $\gamma^2 \approx \exp(4 k \varepsilon)$ factor to the sample complexity is often called the sampling overhead of the QP method.

\subsection{Quasi-Probability Distributions}
The QP sampling procedure used in the examples above might seem somewhat ad-hoc, and we would want our design of an error mitigation protocol to be more automated.
We do this using the framework of quasi-probability distributions.

A probability distribution over a discrete set $X$ can be viewed as a mapping $p: X \rightarrow [0, 1]$, assigning a probability $p_x$ to each event $x \in X$.
The probabilities must be non-negative $\forall x \in X\;\; p_x \geq 0$ and their total must be equal to $\sum_{x\in X} p_x = 1$.

Similarly, we define a {\em quasi-probability distribution} over a discrete set $X$ as a mapping $q: X \rightarrow \R$, assigning a {\em quasi-probability} $q_x$ to each event $x \in X$.
The difference between probabilities and quasi-probabilities is that quasi-probabilities may be negative and are not necessarily bounded by $1$.

This difference is often measured with the {\em negativity} of a quasi-probability distribution which is defined as the sum of its negative values, or using the {\em quasi-probability norm} which is defined as the sum of the absolute values of the quasi-probabilities and is often denoted by $\gamma$ or $W$.
\[
\text{QP Norm} = \gamma \defeq \sum_x \abs{q_x} = 1 + 2\cdot \text{Negativity} \geq 1
\]

If $q$ is a classical probability distribution, then its negativity is $0$ and its norm is $1$.
If $q$ is a quasi-probability, then its negativity and norm can be arbitrarily large.

Quasi-probability distributions are ubiquitous in the study of quantum mechanics (see~\eg~\cite{gherardini2024quasiprobabilities}), but for our purposes, we focus mainly on their application to quantum error mitigation~\cite{temme2017error, piveteau2022quasiprobability}.

Consider the error channels in the examples above.
Each of these channels causes an error $E_i$ with probability $\varepsilon_i$ and does nothing with probability $1 - \varepsilon_i$.
In Section~\ref{sec:PEC}, we show that in some sense, the inverse of these lossy channels can be thought of as quasi-probabilistic channels which causes the errors $E_i$ with quasi-probabilities $-\varepsilon_i$ and do nothing with quasi-probabilities $1+\varepsilon_i$.
These quasi-probability distributions have QP-norm $\gamma_i = 1+2\varepsilon_i$.

In this framework, the QP method performs a Monte-Carlo simulation of precisely this quasi-probability distribution.
Each event (in this case, our choice of the bits $b_i$ that decide whether or not to apply an intentional error) is selected with probability 
\[
p_{\vec{b}} = \frac{\abs{q_{\vec{b}}}}{\gamma} = \prod_{1 \leq i \leq k} \frac{\abs{q_{b_i}}}{\gamma_i} = \prod_i p_{b_i}
\]
and the output is multiplied by a sign and rescaling factor of
\[
\frac{q_{\vec{b}}}{p_{\vec{b}}} = (-1)^{b_1 \oplus \cdots \oplus b_k} \gamma .
\]

This method of performing a Monte-Carlo simulation of quasi-probability distributions is fairly standard in quantum information processing, so its adaptation to quantum error mitigation is natural.

\subsection{Organization of the Paper}
In Section~\ref{sec:trace_reconstruction}, we give a more detailed overview of recent trace reconstruction literature, providing the motivation for and definition of the $k$-mer estimation problem.
In Section~\ref{sec:PEC}, we give a classical description of the quasi-probability method, which is the main tool we use to construct our $k$-mer estimation algorithm in Section~\ref{sec:pec_for_insdel}.
Finally, in Section~\ref{sec:linear_programming} we show that this $k$-mer estimation algorithm yields a time and sample efficient trace reconstruction algorithm.

\section{Trace Reconstruction and \texorpdfstring{$k$}{k}-mer Estimation}\label{sec:trace_reconstruction}
We begin our analysis by giving an overview of recent trace reconstruction algorithms, and how $k$-mer estimation oracles fit in them.

\subsection{Overview of Trace Reconstruction Algorithms}

Chase~\cite{chase2021separating} derived an algorithm for reconstructing any input string $\bfx$ from $\exp\left(\wt{O}(n^{1/5})\right)$ traces, improving on the works of De et al. and Nazarov and Peres~\cite{de2017optimal,nazarov2017trace} who show that $\exp\left(O(n^{1/3})\right)$ traces are necessary and sufficient for ``mean-based'' trace reconstruction algorithms.
In this section we will try to first give a very high level overview of the ideas behind these results using the language of Fourier analysis.

\paragraph{Mean-Based Analyses}
In all $3$ of these results, the analysis is based on bounding the number of traces needed to distinguish between any two strings $\bfx, \bfy \in \{0, 1\}^n$.
In mean-based trace reconstruction, one considers the mean vector $\mu_i^\bfx = \Expect{\wt{\bfx}\sim \mcC(\bfx)}{\wt{\bfx}_i}$ whose $i$th coordinate is the expectation over the $i$th bits of each trace.

Clearly, given sufficiently many traces, we may estimate each coordinate in $\mu$ to a high degree of accuracy.
De et al. and Nazarov and Peres bound the sample complexity of mean-based trace reconstruction by bounding the minimal separation between any two mean vectors $\min_{\bfx \neq \bfy \in \bits^n} \Norm{\mu^\bfx - \mu^\bfy}$.
Linearity gives us $\mu^\bfx - \mu^\bfy = \mu^{\bfx - \bfy}$, reducing their analysis to the question
\[
\min_{\Delta \in \set{0, \pm 1}^n \setminus \set{0^n}} \Norm{\mu^\Delta} = ?
\]

The key idea in this mean-based analysis is to show that $\mu^\Delta$ can be seen as a ``smoothed'' version of this difference $\Delta$.
In other words, each coordinate in $\mu^\Delta$ is the average over a range of entries of $\Delta$.

In the Fourier space, smoothing is equivalent to ``weighing down'' the higher frequencies of the input function. Defining $p_\Delta(z) = \sum_i \Delta_i z^i$, De et al. and Nazarov and Peres use a series of complex analysis results due to Borwein and Erd\'{e}lyi~\cite{borwein1997littlewood} and Borwein et al.~\cite{borwein1999littlewood} to show that:
\begin{equation}
    \label{eq:borwein_erdlyi}
    \min_{\Delta \in \set{0, \pm 1}^n \setminus \set{0^n}} \set{\max_{\theta \in [-\alpha, \alpha]} \abs{p_\Delta \left(e^{i\theta}\right)}} = \exp \left( - \Theta\left( \alpha^{-1} \right) \right)
\end{equation}

Combined with the smoothing which weighs down the coefficient of a frequency $\alpha$ by a factor of
\[
\abs{\frac{p_{\mu^\Delta}(e^{i\alpha})}{p_{\Delta}(e^{i \alpha})}} = \left(1 -\Theta(\alpha^2)\right)^n = \exp\left(-\Theta(\alpha^2 n)\right)
\]
they achieve the best separation when setting $\alpha = n^{-1/3}$, yielding a formula of the form
\[
\min_{\Delta \in \set{0, \pm 1}^n \setminus \set{0^n}} \Norm{\mu^\Delta} \approx \min_{\Delta \in \set{0, \pm 1}^n \setminus \set{0^n}} \set{\max_{\theta \in [-\pi, \pi]} \abs{p_{\mu^\Delta}(e^{i\theta})}} = \exp \left( -\Theta(n^{1/3}) \right)\,,
\]
where the $\approx$ may suppress a $\poly(n)$ factor.

\paragraph{Sparsifying the Input String}

A general rule of thumb with Fourier transformations, is that the sparser the input signal, the more the Fourier weights concentrate on the lower frequencies.
In particular, Chase refines equation~\eqref{eq:borwein_erdlyi}, proving that if the minimization is limited to only ``sparse'' vectors $\Delta$, then $\abs{p_\Delta \left(e^{i\theta}\right)}$ obtains much larger values.

In order to produce a sparse signal, Chase combines two ideas.
The first is a Lemma by Robson (see Lemma~\ref{lem:sparse_substring}), which promises that for any two strings $\bfx$ and $\bfy$ of length $n$, we can find a marker $\bfw$ of length $k$, for which the set of indices where $\bfw$ appears in $\bfx$ and $\bfy$ differ and are both sparse.
In other words, for some choice of $\bfw \in \bits^k$, the signal $\Delta = \left(1_{\bfx_{i:i+k} = \bfw} - 1_{\bfy_{i:i+k} = \bfw} \right)$ is sparse and non-zero.

\begin{definition}[Periods of a String]
Let $\bfw \in \bits^k$ be a string.
We say that $\bfw$ has period $p$ for some $1 \leq p \leq k$, if $\bfw_{:k-p} = \bfw{p:}$.
\end{definition}

\begin{lemma}[Lemmas 1 and 2 of \texorpdfstring{\cite{robson1989separating}}{Rob89}]
\label{lem:sparse_substring}
Let $S\in \{0, 1\}^k$ be a string of bits.
\begin{enumerate}
    \item If $S\mid 0$ (defined as the concatenation of $0$ to the end of $S$) has period $\leq k/2$, then $S\mid 1$ has a period of at least $k/2$.
    \item If $S$ has period $p$ and a string $U$ contains two occurrences of $S$ starting at $U_i$ and at $U_j$, then $\abs{i-j} \geq p$
\end{enumerate}
\end{lemma}

The final component needed for Chase's algorithm is a method of extracting a smoothed version of this sparse signal $\Delta$.
This is where the $k$-mer estimation problem comes in.

\subsection{The \texorpdfstring{$k$}{k}-mer Estimation Problem}

Mazooji and Shomorony~\cite{mazooji2024substring} define the $k$-mer estimation problem as the problem of estimating a ``smoothed'' version of the indicator vector of the indices where a given substring $\bfw$ appears in the input string $\bfx$.
However, their definition depends on the specific choice of a smoothing kernel and this choice is itself dependent on the specific noisy channel.

We deviate slightly from their definition, and instead define $k$-mer estimation as the task of estimating the low frequencies of the indicator set of the appearances of the string $\bfw$ in the input string $\bfx$.
More formally,

\begin{definition}[The \texorpdfstring{$k$}{k}-mer Estimation Problem]
    \label{def:kmer}
    We say that an algorithm $\mcA$ solves the $k$-mer estimation algorithm for error channel $\mcC$ frequency $\alpha \in [-\pi, \pi]$, up to error $\varepsilon$, sample complexity $\NumSamples$ and failure probability $\delta$, if
    \[
    \forall \bfx \in \bits^n \;\;\;\; \Prob{\wt{\bfx}^1, \ldots, \wt{\bfx}^\NumSamples \sim \mcC(\bfx)}{\abs{\mcA\left(\bfw, \wt{\bfx}^1, \ldots, \wt{\bfx}^\NumSamples\right) - p_{\left( 1_{\bfx_{i:i+k}=\bfw} \right)_{i \in [n-k]}}(e^{i\alpha})} < \varepsilon} > 1 - \delta
    \]
\end{definition}

We will measure the effectiveness of $k$-mer estimation algorithms by the scaling of their sample complexity as a function of the other parameters.
Our goal will be to construct a $k$-mer estimation algorithm with a ``good'' sample complexity:

\begin{definition} [Good \texorpdfstring{$k$}{k}-mer Estimation Algorithm]
    We say that a $k$-mer estimation algorithm $\mcA$ has a {\em good sample complexity} if
    \[
\NumSamples = \underbrace{\left(1 + O\left(\alpha^2\right) \right)^n}_{\textnormal{smoothing overhead}} \times \underbrace{\poly\left(\varepsilon^{-1}\right)}_{\textnormal{mean estimation cost}} \times \underbrace{\polylog\left(1 / \delta\right)}_{\textnormal{failure probability overhead}}
\]
\end{definition}

This is the sample complexity achieved by leading $k$-mer estimation algorithms and we do not expect to directly improve upon it (in particular, any significant improvement would also imply a better sample complexity for the trace reconstruction problem), justifying the label ``good''.
This leads us to the main technical theorem we will prove in this paper

\begin{theorem}[General \texorpdfstring{$k$}{k}-mer Estimation Algorithm]
\label{thm:kmer_estimation}
    For any set of channels $\mcC_1, \ldots, \mcC_L$, each of which may be either a deletion, insertion or symmetry channel, there exists a good $k$-mer estimation algorithm for the composite channel $\mcC = \mcC_1 \circ \cdots \circ \mcC_L$.
\end{theorem}

We will construct the algorithm for Theorem~\ref{thm:kmer_estimation} using the quasi-probability method.
Before explaining why this method works, we can already present the algorithms that it will help us derive. 
At first, the design of Algorithms~\ref{alg:qp_single_insdel} and~\ref{alg:qp_kmer_estimation} will seem odd, and it might not be immediately clear why they should work.

In Section~\ref{sec:PEC} we introduce the QP method, and in Section~\ref{sec:pec_for_insdel}, we will use it to derive Algorithms~\ref{alg:qp_single_insdel} and~\ref{alg:qp_kmer_estimation} and prove that they work as intended.

For now, we note the relative simplicity of the Monte Carlo method in Algorithm~\ref{alg:qp_kmer_estimation} -- especially when compared with the complex analysis used by Chase~\cite{chase2021separating} and the numerical differentiation introduced by Rubinstein to deal with insertion channels~\cite{rubinstein2022average}.

\begin{algorithm}[!ht]
\label{alg:qp_single_insdel}

\KwIn{An oracle $f:\bits^k \rightarrow [-C, C]$, an error channel $\mcC$, an accuracy parameter $\varepsilon$}
\KwOut{An oracle $f_{\QP(\mcC)} : \bits^{k^\prime} \rightarrow [-C^\prime, C^\prime]$, such that $C^\prime = 2^{O(k)} \poly(\varepsilon^{-1}) C$ and $\abs{\Expect{\bfx \sim \mcX}{f(\bfx)} - \Expect{\substack{\bfx \sim \mcX \\ \wt{\bfx} \sim \mcC(\bfx)}}{f_{\QP(\mcC)}(\wt{\bfx})}} < \varepsilon C$}

\BlankLine

Set $k^\prime \coloneqq \begin{cases}
    k + 4 \cdot \frac{\delta}{1-2\delta} k + 4 \cdot \frac{\delta}{1 - \delta} \Paren{\log\Paren{\frac{1}{1 - 2\delta}} k + \log\left(\varepsilon^{-1}\right)} & \mcC = \DelChannel_\delta\\
    k & \textup{otherwise}
\end{cases}$ \;

\SetKwProg{Fn}{\texttt{def} $f_{\QP(\mcC)}\left(\wt{\bfx}\right)$ \texttt{:}}{}{}

\Fn{}{
\If{$\mcC = \DelChannel_\delta$}{
    set $\gamma = \frac{1}{1 - 2\delta}$ \;
    
    \For{$i \gets 1$ \KwTo $k$}{
        draw $j_i \sim \Geom\left(1 - \frac{\delta}{1 - \delta} \right)$ \;
        delete $j_i$ bits of $\wt{\bfx}$ from the $i$th bit onwards\;
        set $s_i = (-1)^{j_i}$ \;
        \If{$j_1 + \cdots + j_i > k^\prime - k$}{
            Halt and return 0 \;
        }
    }
}

\If{$\mcC = \InsChannel_\eta$}{
    set $\gamma = \frac{1+\eta}{1-\eta}$ \;
    
    \For{$i \gets 1$ \KwTo $k$}{
        draw $j_i \sim \Ber\left( \frac{\eta}{1-\eta} \right)$ \;
        insert $j_i$ iid uniformly distributed bits before the $i$th bit of $\wt{\bfx}$ \;
        set $s_i = (-1)^{j_i}$ \;
    }
}

\If{$\mcC = \textup{Symmetry Channel}_\sigma$}{
    set $\gamma = \frac{1}{1-2\sigma}$ \;
    
    \For{$i \gets 1$ \KwTo $k$}{
        draw $j_i \sim \Ber\left( \frac{\sigma}{1-2\sigma} \right)$ \;
        set $\wt{\bfx}_i \coloneqq \wt{\bfx}_i \oplus j_i $ \;
        set $s_i = (-1)^{j_i}$ \;
    }
}

Output $f\left(\wt{\bfx}\right) \times \gamma^k \times \prod_{i = 1}^k s_i$

}

\caption{QP Decomposition of a Simple Channel}
\KwRet $f_{\QP(\mcC)}$ \;
\end{algorithm}

\begin{remark}

When the error channel is a deletion channel, the QP estimator $f_{\QP(\mcC)}$ defined in Algorithm~\ref{alg:qp_single_insdel} may delete some of the bits of its input, so if it were implemented naively, this could cause its output to depend on more than the first $k$ bits of the input (which would violate the claim that $f_{\QP(\mcC)}$ depends on only the first $k^\prime$ bits of its input -- see Lemma~\ref{lem:qp_single_insdel}). 

By setting $k^\prime$ to be sufficiently large compared to $k + \frac{\delta}{1 - 2\delta} k$ (which is the expected number of bits of the input read by $f_{\QP(\mcC)}$) in line $1$ of Algorithm~\ref{alg:qp_single_insdel}, we ensure that the buffer length is exceeded only with vanishingly small probability (see Claim~\ref{clm:no_overflows}).
Lines $9$ and $10$ of the algorithm check that we indeed did not have such a ``buffer overflow'', and the condition is triggered with a negligible probability, resulting in a bias of at most $\varepsilon$ in the estimator.

\end{remark}

\begin{algorithm}[!ht]
\label{alg:qp_kmer_estimation}

\KwIn{Error channel $\mcC = \mcC_1 \circ \cdots \circ \mcC_L$, traces $\wt{\bfx}^1, \ldots, \wt{\bfx}^\NumSamples$, frequency $\omega \in [-\pi, \pi]$, marker $\bfw \in \bits^k$}
\KwOut{\( k \)-mer Estimator $E \approx \sum_{j \in [n]} e^{ij\omega} 1_{\bfx_{j:j+k} = \bfw}$}

\BlankLine

Define the oracle $f^0: \bits^k \rightarrow [-1, 1]$ to be:
\[
f^0(\bfx) = 1_\bfw = 
\begin{cases} 
1 & \bfx = \bfw \\ 
0 & \textup{otherwise} 
\end{cases}
\]

\For{$j \gets 1$ \KwTo $L$}{
    \If{$\mcC_j = \DelChannel_\delta$ with $\delta \geq 1/2$}{
        decompose $\mcC_j$ into deletion channels $\mcC_j = \mcC_{j, 1} \circ \cdots \circ \mcC_{j, O(\log((1-\delta)^{-1}))}$ with deletion probability $\delta^\prime < 1/3$\;
    }
    Apply the QP Decomposition Algorithm~\ref{alg:qp_single_insdel} to update the oracle $f^j \coloneqq f^{j-1}_{\QP(\mcC_j)}$\;
}

\For{$j \gets 1$ \KwTo $n$}{
    Compute the empirical average:
    \[
    \mu_j = \frac{1}{\NumSamples} \sum_{t \in [\NumSamples]} f^L(\wt{\bfx}^t_{j:})
    \]
}

Let $\zeta = e^{i \omega}$ and $z = G^{-1}_{\abs{\mcC}}(\zeta)$ (as defined in Section~\ref{subsec:generating_funcs}).

\Return $E = \zeta^{-1} \sum_{j \in [n]} \mu_j z^j$

\caption{QP-based $k$-mer Estimation Algorithm}
\end{algorithm}

\section{The Quasi-probability Method}\label{sec:PEC}
In this section we will give a detailed overview of the Quasi-Probability (QP) method, also known as Probabilistic Error Cancellation (PEC).
While this method was initially designed to be a method for simulating noise-free quantum computations on noisy quantum devices~\cite{temme2017error} and is often presented using terminology from quantum information theory, we show that its main ideas translate nicely into classical computing and no background in quantum is needed to understand this section.
Where relevant, we will mention the relationship between our classical description and the quantum description of this method.

\subsection{High Level Idea}

PEC takes as input a description of some ideal process we want to simulate and of the inherent noise in the system we use to simulate it.
As our main example, we imagine that the description is some classical circuit we wish to run on a classical computer, that our goal is to find the first bit of its output (or its expectation if the circuit is not deterministic), but that we are given access only to a faulty computer.
 
Let $\Lambda$ be the lossy noise channel applied to our state before each step of this process.
We view $\Lambda$ as an $\NumStates\times \NumStates$ stochastic matrix\footnote{In the quantum setting, $\Lambda$ would be a ``superoperator'', defined by a $\NumStates^2 \times \NumStates^2$ complex matrix.}, representing one Markovian step on all $\NumStates$ possible states of the process.
In our circuit example, $\NumStates=2^n$ represents the number of possible computational states, and $\Lambda$ might be a symmetry channel being applied to some bits of the state.
 
Conceptually, we would like to correct our process by inserting intermediate steps which somehow apply the inverse channel $\Lambda^{-1}$ to our state.
However, $\Lambda^{-1}$ is ``non-Physical'', meaning that such a set of operations cannot exist.
This is because $\Lambda$ is lossy, implying that $\Lambda^{-1}$ must somehow add more information into the system.
 
To overcome this, we introduce a concept called {\em quasi-probability distributions}.
We define a quasi-probability (QP) distribution to be a map $q: U \rightarrow \R$ which assigns a real value $q_u$ to each member $u$ of some discrete set of events $U$.
Standard probability distributions $p$ are defined with the constraints that their total weight $\sum_u p_u = 1$ is $1$ and that each entry is non-negative $\forall u\;\;p_u \geq 0$.
On the other hand, quasi-probability distributions relax the second condition allowing for negative weights $q_u < 0$.

Quasi-probability distributions can produce ``impossible'' channels like $\Lambda^{-1}$ because they do not adhere to conventional information theory.
We utilize this by first decomposing our desired channel $\Lambda^{-1}$ into a quasi-probability distribution of ``possible'' channels, and then use a Monte-Carlo simulation to estimate the expectation of properties of the resulting quasi-probabilistic process.

\subsection{Quasi-Probability Decomposition} 
The first step in the QP method is called a ``quasi-probability decomposition’’ \cite{piveteau2022quasiprobability}.

In this step we search for a set of operations represented by stochastic matrices $L_1, \ldots, L_k$ and real coefficients $q_1, \ldots, q_k$, such that as matrices $\Lambda^{-1} = \sum_i q_i L_i$.
The operations $L_i$ are called the ``decomposition set’’ and the coefficients $q_i$ represent a quasi-probability distribution over the choices of $L_i$.

This step is the only part of the QP method which is not constructive, and each error channel might require some ingenuity for selecting its decomposition set and QP distribution.
Typically, the set of operations in our decomposition set will be similar to the set of errors the channel can cause itself.
 
For instance, consider a symmetry channel $\Lambda = (1 - \sigma) I + \sigma X_1$ (where $X_1$ is the symmetry operation on the first bit). This channel flips the first bit with probability $\sigma < 1/2$.
 A simple calculation gives us that
\[
\Lambda^{-1} = \frac{1-\sigma}{1-2\sigma} I - \frac{\sigma}{1-2\sigma} X_1
\]

In fact, for any error channel $\Lambda$ which does nothing with probability $1 - \varepsilon$ and applies some general error $E$ with small probability $\varepsilon$, if we are allowed to neglect $O(\varepsilon^2)$ terms, we can easily show that $\Lambda^{-1} \approx (1+\varepsilon)I - \varepsilon E$.

Finally, when applying the QP method to insertion and deletion channels, we will show that they have ``good'' QP decompositions into insert-$k$ and delete-$k$ channels.

\subsection{Sampling from a QP Distribution}
After selecting a QP decomposition, the next step is to use it to compute the expectation value of some {\em observable} -- any function from the set of states of our process to the real numbers $M : [\NumStates] \rightarrow \R$.
The observable can also be viewed as a real vector $\vec{M} \in \R^\NumStates$.

Let $\rho \in \R^\NumStates$ be a vector representing the distribution over possible states of our system immediately after the channel $\Lambda$.
Morally, our goal is to apply some operation on the system so that its output will behave like $\rho^\prime = \Lambda^{-1} (\rho)$.

We do this by constructing the channel $\wt{\Lambda}$.
In principle, we want this channel to sample each process $L_i$ in the QP set with probability $q_i$.
If this were possible, then we would have the following equality when viewing the channels as transition probability matrices:
\[
\wt{\Lambda} = \sum_i q_i L_i = \Lambda^{-1}
\]

The problem with this approach is that $\vec{q}$ might not be a probability distribution (\ie~it may contain negative values).
To overcome this, we sample each $L_i$ with probability $p_i \propto \abs{q_i}$, and retain a factor of $f_i = p_i / q_i$.

For any observable $M$ mapping states of the process to real numbers, we simply multiply its output by the retained factor $f_i$, to obtain
\[
\langle M \rangle_{\QP} \defeq \E_{\substack{i\leftarrow p \\ x\leftarrow \rho \\ x^\prime \leftarrow L_i(x)} } \left[f_i M\left(x^\prime\right)\right] = \left\langle \vec{M}, \sum_i p_i f_i L_i \rho \right\rangle = \left\langle \vec{M}, \Lambda^{-1} \rho \right\rangle = \langle M \rangle_{\Lambda^{-1}}\,,
\]
where $\langle M \rangle_{\QP}$ is the expected output of the sampling procedure above, and $\langle M \rangle_{\Lambda^{-1}} \defeq \left\langle \vec{M}, \Lambda^{-1} \rho \right\rangle$ is the expectation of the observable $M$ on the output of the quasi-probabilistic channel that we want to estimate.

This means that the average over samples from the QP sampling is indeed equal to the average of our desired channel $\Lambda^{-1}$.
Naturally, we may apply QP recursively to invert multiple independent lossy channels\footnote{Known as a ``Markovian model'' in the language of quantum channels.} $\Lambda_1, \ldots, \Lambda_V$ spread throughout the noisy process.
To perform the QP method recursively, we simply insert the operation $L_{j, i}$ before the $j$th noisy step of our process independently w.p. $p_{j, i}$, and multiply our output by the product of all the factors $\prod_j f_{j, i}$.

\subsection{Sampling Overhead}

Note however, that there is no ``free lunch''.
We pay for the ability to simulate these ``non-Physical'' channels $\Lambda_j^{-1}$, through an increase in the variance of our samples. 

This is due to the factors $f_{j, i} = q_{j, i} / f_{j, i} = \text{sign}(q_{j, i}) \sum_{i^\prime} \abs{q_{j, i^\prime}}$.
We set $\sum_i q_i = 1$, but some of the $q_i$s are negative, so $\gamma_j \defeq \sum_i \abs{q_i^j}$ (called the QP norm) is greater than $1$.
This means that even if $\left \Vert M \right \Vert_\infty = C$ is a constant, the variance of our estimation of the expectation of $M$ through a PEC process can scale at most like
\[
\Var(M_\QP) \leq C^2 \prod_j \gamma_j ^ 2
\]

Returning to our example of running a circuit in the presence of noise, if we have $V$ gates in the circuit, and we assume each gate $G_j$ incurs a small probability $\varepsilon_j$ of some known error process $E_j$ being applied to the system, we can use PEC to correct the output of this circuit (to within an error of order $O(\sum_j \varepsilon_j^2)$) using the QP-decomposition $\Lambda_j^{-1} \approx (1+\varepsilon_j) I - \varepsilon_j E_j \pm O(\varepsilon_j^2)$.

Assuming $\left \Vert M \right \Vert_\infty = 1$, the variance of an estimate $M^*$ made by averaging over $\NumSamples$ samples of PEC would be
\[
\Var (M^*) = \Var(M_{\QP}) / \NumSamples \leq \prod_j \left(\abs{1+\varepsilon_j} + \abs{-\varepsilon_j}\right)^2 / \NumSamples \approx \exp\left(4\sum_j \varepsilon_j\right) / \NumSamples
\]

This means that in order to get a constant error bar on our estimate $M^*$, we may need to use $\NumSamples \approx \exp\left(4\sum_j \varepsilon_j\right)$ PEC samples.

Note that this scaling is tight in some cases.
For instance, if the ideal computation we wanted to simulate produced some distribution on bit-strings $w\in \{0, 1\}^n$, our goal was to compute the expectation of $x = w_1 \oplus \cdots \oplus w_n$ and the noise channel was an \iid~binary symmetry channel with parameter $\varepsilon$ applied to each of the output bits independently.
In this case, if $x$ had a bias $b$, then in the presence of noise its bias would degrade by a factor of $(1-2\varepsilon)^n$, and we would need $\Theta\left(\exp(4\varepsilon n) \delta^{-2}\right)$ samples to reconstruct it to within an additive error of $\delta$.

\subsection{Formal Statement and Proof}

\begin{definition}[Quasi-Probability Decomposition]
Let $\Lambda \in \R^{\NumStates \times \NumStates}$ be a stochastic matrix representing the Markov process of some lossy error channel.

We say that $\set{(q_j, \Lambda_j)}_{j \in [\ell]}$ is $\varepsilon$ away from being a {\em quasi-probability decomposition of $\Lambda^{-1}$}, with norm $\gamma \defeq \sum_{j = 1}^\ell \abs{q_j}$, if:
\begin{itemize}
    \item $q_1, \ldots, q_\ell \in \R$ are scalars and $\sum_{j = 1}^{\ell} q_j = 1$.
    \item $\Lambda_1, \ldots, \Lambda_\ell \in \R^{\NumStates \times \NumStates}$ are stochastic matrices.
    \item $\Norm{\Lambda^{-1} - \sum_{j = 1}^\ell q_j \Lambda_j}_{1, \infty} < \varepsilon$, where $\Norm{A}_{1, \infty} = \max_{i \in [\NumStates]} \set{\sum_{j \in [\NumStates]} \abs{A_{i, j}}}$ denotes the maximum over rows of $A$ of their $\ell_1$ norms.
\end{itemize}
\end{definition}

It is not always clear how to construct a ``good'' QP decomposition.
In Section~\ref{sec:pec_for_insdel}, we derive QP decompositions for several synchronisaiton channels, but first, we conclude with the main result of this section -- proving that a QP decomposition yields an estimator:

\begin{claim}
\label{clm:sample_qp}
Let $\Lambda$ denote the stochastic matrix corresponding to a lossy error channel, and let $(q_1, \Lambda_1), \ldots, (q_\ell, \Lambda_\ell)$ be $\varepsilon$ far from being a QP decomposition of $\Lambda^{-1}$ with norm $\gamma$.

Then for any function $f:[\NumStates] \rightarrow [-C, C]$, there exists a function $f_{\QP(\Lambda)}:[\NumStates] \rightarrow [-C^\prime, C^\prime]$, such that
\begin{itemize}
    \item $f_{\QP(\Lambda)}$ makes one call to one of the channels $\Lambda_1, \ldots, \Lambda_\ell$, and one call to an oracle that computes $f$.
    \item The output of $f_{\QP(\Lambda)}$ is bounded by $C^\prime = \gamma C$.
    \item The expectation of $f_{\QP(\Lambda)}$ on outputs of the noisy channel is close to the expectation on its input. In other words, for any distribution $\mcX$ on the state space $[\NumStates]$, we have
    \[
    \abs{\Expect{x \sim \mcX}{f(x)} - \Expect{\wt{x} \sim \Lambda(\mcX)}{f_{\QP(\Lambda)}(\wt{x})}} < \varepsilon C
    \]
\end{itemize}
\end{claim}

\begin{proof}[Proof of Claim~\ref{clm:sample_qp}]

We define the oracle $f_{\QP(\Lambda)}$ when run on the input $x$ to sample $j \in [\ell]$ from the distribution
\[
p_j = \frac{\abs{q_j}}{\sum_{i \in [\ell]} \abs{q_i}}\,,
\]
apply the Markov process $\Lambda_j$ on $x$ and then output
\[
f_{\QP(\Lambda)}(x) = \frac{q_j}{p_j} f\left(\Lambda_j(x)\right)\,.
\]

It's outputs are bounded by
\[
\abs{f_{\QP(\Lambda)}(x)} = \abs{\frac{q_j}{p_j}} \abs{f\left(\Lambda_j(x)\right)} \leq \sum_{j \in [\ell]} \abs{q_j} \cdot C = \gamma C = C^\prime\,.
\]

Finally, we have
\[
\EE{f_{\QP(\Lambda)}(x)} = \Expect{j \sim p_j}{\frac{q_j}{p_j} f\left(\Lambda_j(x)\right)} = \sum_{j = 1}^\ell p_j \frac{q_j}{p_j} \EE{f\left(\Lambda_j(x)\right)} = \sum_{j = 1}^\ell q_j \EE{f\left(\Lambda_j(x)\right)}\,,
\]
and from our assumption that $(q_1, \Lambda_1), \ldots, (q_\ell, \Lambda_\ell)$ is $\varepsilon$ far from being a QP decomposition of $\Lambda^{-1}$, we have
\begin{align*}
    \abs{\sum_{j = 1}^\ell q_j \EE{f\left(\Lambda_j(x)\right)} - \EE{f(x)}} &= \abs{\IP{\vec{f}}{\Paren{\Lambda^{-1} - \sum_{j \in [\ell]} q_j \Lambda_j}^\intercal \vec{\mcX}}} \leq\\
    &\leq \Norm{f}_\infty \Norm{\Paren{\Lambda^{-1} - \sum_{j \in [\ell]} q_j \Lambda_j}^\intercal \vec{\mcX}}_1 < \varepsilon C\,,
\end{align*}
where $\vec{f} = \Paren{f(1), \ldots, f(\NumStates)}$ and $\vec{\mcX}$ is the probability density function of $\mcX$.

\end{proof}

\section{Quasi-probability Decomposition of Insertion-Deletion Channels}\label{sec:pec_for_insdel}
In the previous sections, we introduced the trace reconstruction problem and the quasi-probability method.
In this section, we will apply the quasi-probability method to the error channels of the trace reconstruction problem.
In particular, we will use the quasi-probability method to derive algorithms for $k$-mer estimation and population recovery in the large support regime (proving Theorems~\ref{thm:kmer_estimation} and~\ref{thm:pop_recovery}).

Throughout this section, our goal will be to find a QP inversion of any combination of insertion and deletion channels.
Before approaching this general setting, we will design QP decompositions of individual insertion and deletion channels.
Moreover, we will limit this part of the analysis to the low deletion probability regime $\delta < 1/2$.

In particular, we will prove the following lemma:

\begin{lemma}
    \label{lem:qp_single_insdel}
    Let $\mcC$ be an error channel that is either a symmetry channel with parameter $\sigma < 1/2$, a deletion channel with deletion rate $\delta < 1/2$, or an insertion channel with insertion rate $\eta < 1$, and let $f:\bits^k \rightarrow [-C, C]$ be any bounded mapping from a length $k$ prefix of the input string to a scalar.
    Let $\varepsilon > 0$ be an accuracy parameter.

    Then, there exists a QP estimator $f_{\QP(\mcC)} : \bits^{k^\prime} \rightarrow [-C^\prime, C^\prime]$ such that
    \begin{itemize}
        \item $f_{\QP(\mcC)}$ takes as input a string of length $k^\prime = O(k + \log(\varepsilon^{-1}))$
        \item Its output is bounded in absolute value by $C^\prime = C \exp\Paren{O(k + \log(\varepsilon^{-1}))}$
        \item For any distribution of input strings $\mcX$, running $f_{\QP(\mcC)}$ on the prefix of the trace of a sample $\bfx \sim \mcX$ is a nearly unbiased estimator of the expectation of $f(\bfx_{:k})$:
        \[
        \abs{\Expect{\bfx \sim \mcX}{f\left(\bfx_{:k} \right)} - \Expect{\substack{\bfx \sim \mcX \\ \wt{\bfx} \sim \mcC(\bfx)}}{f_{\QP(\mcC)}\left(\wt{\bfx}_{:k^\prime} \right)}} \leq \varepsilon
        \]
    \end{itemize}
\end{lemma}

For the rest of this section, we adopt the notation that when $f$ is a function that takes as input $k$ bits and $\bfx$ is a string of length $n\neq k$, we set $f(\bfx)$ to be the application of $f$ to the first $k$ bits of $\bfx$ if $n > k$, and return $0$ if $n < k$.

\subsection{Composability of the QP Decomposition -- Lemma~\ref{lem:qp_single_insdel} Suffices}

Before proving Lemma~\ref{lem:qp_single_insdel}, we show that its self-composability suffices to extend it to any combination of insertion deletion channels.
Indeed, it is easy to see that for any set of channels $\mcC_1, \ldots, \mcC_L$, we can apply Lemma~\ref{lem:qp_single_insdel} recursively, to construct $f_{\QP(\mcC_1 \circ \cdots \circ \mcC_L)} = f_{\QP(\mcC_L), \ldots, \QP(\mcC_1)}$, and that composing the guarantees of Lemma~\ref{lem:qp_single_insdel} using the triangle inequality gives us
\begin{equation}
    \begin{aligned}
        &\abs{\Expect{\bfx \sim \mcX}{f(\bfx)} - \Expect{\substack{\bfx \sim \mcX \\ \wt{\bfx} \sim \mcC_1 \circ \cdots \circ \mcC_L(\bfx)}}{f_{\QP(\mcC)}\left(\wt{\bfx}\right)}} \leq \\
        &\;\; \leq \sum_{1 \leq i \leq L}  \abs{\Expect{\substack{\bfx \sim \mcX \\ \wt{\bfx} \sim \mcC_{i+1} \circ \cdots \circ \mcC_L(\bfx)}}{f_{\QP(\mcC_{L:i+1})}\left(\wt{\bfx}\right)} - \Expect{\substack{\bfx \sim \mcX \\ \wt{\bfx} \sim \mcC_i  \circ \cdots \circ \mcC_L(\bfx)}}{f_{\QP(\mcC_{L:i})}\left(\wt{\bfx}\right)}} \leq L \varepsilon
    \end{aligned}
\end{equation}

This yields the following Corollary:

\begin{corollary}
    \label{cor:composed_insdel}
    Let $\mcC_1, \ldots, \mcC_L$ be any set of channels individually covered by Lemma~\ref{lem:qp_single_insdel}.
    Then, there exists a QP estimator for the composite channel $\mcC = \mcC_1 \circ \cdots \circ \mcC_L$.
\end{corollary}

Finally, in order to extend our results to the high deletion probability regime, we simply note that any deletion channel can be decomposed into a series of deletion channels with lower deletion probabilities.
Indeed, if $\mcC_{\delta_1}, \mcC_{\delta_2}$ are the deletion channels with deletion probabilities $\delta_1, \delta_2$, then it is easy to see that their composition $\mcC_{\delta_1} \circ \mcC_{\delta_2}$ is the deletion channel with deletion probability $\mcC_{1 - (1 - \delta_1) (1 - \delta_2)}$.

Therefore, for any $\delta \in (0, 1)$, we can find a pair $\delta^\prime < 1/2$ and $L \in \N$ such that $1 - \delta = (1 - \delta^\prime)^L$.
Therefore,
\[
\mcC_{\delta} = \underbrace{\mcC_{\delta^\prime} \circ \cdots \circ \mcC_{\delta^\prime}}_{\times L}
\]
and we can extend Lemma~\ref{lem:qp_single_insdel} to this deletion probability as well.

Finally, we note that Corollary~\ref{cor:composed_insdel} implies Theorem~\ref{thm:pop_recovery}:

\begin{proof}[Corollary~\ref{cor:composed_insdel} implies Theorem~\ref{thm:pop_recovery}]
Let $\mcC = \mcC_1 \circ \cdots \circ \mcC_L$ be any combination of insertion, deletion and symmetry channels, and let $\mcX$ be any population of strings of length $k$.
From the argument above, we may assume that the deletion channels in this combination of channels have deletion probabilities below $1/2$ (otherwise, we decompose each such deletion channel into multiple deletion channels each with a lower deletion probability).

If we can estimate the probability of drawing any single string $\bfy \in \set{0,1}^k$, to within an error of $\varepsilon^\prime = 2^{-k} \varepsilon$, then we will have recovered the population to within a TVD of at most $\varepsilon$.

Let $\bfy \in \set{0, 1}^k$ be some bitstring, and consider the indicator mapping
\[
f = 1_\bfy:\set{0, 1}^k \rightarrow [-1, 1]\;\;\;\; f(\bfx) = \begin{cases}
    1&\bfx=\bfy\\
    0&\bfx\neq\bfy
\end{cases}\,.
\]
Corollary~\ref{cor:composed_insdel} allows us to construct a QP estimator for the composite channel $\mcC$.
In other words, there exists a function $f_{\QP(\mcC)}$ that takes as input strings of length
\[
k^\prime = O(k + \log((\varepsilon^\prime)^{-1})) = O(k + \log(\varepsilon^{-1}))\,,
\]
has bounded output
\[
\Norm{f_{\QP(\mcC)}}_\infty = 2^{O(k + \log(\varepsilon^{-1}))}\,,
\]
and for any distribution of input strings, running $f_{\QP(\mcC)}$ on the prefix of the trace of a sample $\bfx \sim \mcX$ is a nearly unbiased estimator of the expectation of $f(\bfx_{:k})$:
\[
\abs{\Expect{\bfx \sim \mcX}{f\left(\bfx_{:k} \right)} - \Expect{\substack{\bfx \sim \mcX \\ \wt{\bfx} \sim \mcC(\bfx)}}{f_{\QP(\mcC)}\left(\wt{\bfx}_{:k^\prime} \right)}} \leq \frac{1}{2} \varepsilon^\prime\,.
\]

Moreover, note that due to our choice of $f$ to be the indicator function for the string $\bfy$, we have
\[
\Expect{\bfx \sim \mcX}{f\left(\bfx_{:k} \right)} = \Prob{\bfx \sim \mcX}{\bfx = \bfy}\,.
\]

Let $\mu_\bfy$ be the empirical mean of this $f_{\QP(\mcC)}(\wt{\bfx})$ over $\NumSamples = \frac{10 k}{(\varepsilon^\prime)^2} \Norm{f_{\QP(\mcC)}}_\infty^2 = 2^{O(k + \log(\varepsilon^{-1})}$ iid traces $\wt{\bfx}^1, \ldots \wt{\bfx}^\NumSamples \sim \mcC(\mcX)$.
From the Hoeffding inequality, with very high probability this empirical mean is a good estimate of the mean of the QP estimator.
In other words,
\[
\Prob{}{\abs{\mu_\bfy - \Expect{\substack{\bfx \sim \mcX \\ \wt{\bfx} \sim \mcC(\bfx)}}{f_{\QP(\mcC)}\left(\wt{\bfx}_{:k^\prime} \right)}} > \frac{1}{2}\varepsilon^\prime} \ll 2^{-k}\,.
\]

Therefore, from the triangle inequality,
\begin{equation}
    \label{eq:proof_1_2_high_prob}
    \Prob{}{\abs{\mu_\bfy - \Prob{\bfx \sim \mcX}{\bfx = \bfy}} > \varepsilon^\prime} =\Prob{}{\abs{\mu_\bfy - \Expect{\bfx \sim \mcX}{f\left(\bfx_{:k} \right)}} > \varepsilon^\prime} \ll 2^{-k}
\end{equation}

We apply the process above for each string $\bfy \in \set{0, 1}^k$, resulting in a set of estimates $\mu_\bfy$ for $\Prob{\bfx \sim \mcX}{\bfx = \bfy}$.
Applying the union bound to equation~\eqref{eq:proof_1_2_high_prob}, we conclude that with high probability, all of these estimates are correct up to an additive error of $\pm \varepsilon^\prime$, and therefore constitute a recovery of $\mcX$ to within a TVD of $\varepsilon$.

\end{proof}

\subsection{Conventions for Synchronization Channels}

The final step before we prove Lemma~\ref{lem:qp_single_insdel}, is that we set a few conventions that will help us avoid minor technical issues with matching indices of bits in the input and output strings.

The deletion channel is defined as the process which deletes each bit of the input string with probability $\delta \in (0, 1)$.
This formulation does not imply a specific order in which bits are deleted or retained.
However, for our purpose, it will be beneficial to express this channel as a random process.

In particular, we define the deletion channel as a process which goes over the input string {\em in reverse} (\ie~from the last bit of the input to the first), and for each bit, deletes it w.p. $\delta$.
In other words, we define $\Del{i}{j}$ to be the operation which deletes $j$ bits from the $i$th bit of the input onwards ($\Del{i}{j}(\bfx) = \bfx_{:i} \bfx_{i+j:}$), and we define the deletion channel as
\[
\DelChannel_\delta = \Del{1}{\Bernoulli(\delta)} \circ \Del{2}{\Bernoulli(\delta)} \circ \cdots
\]
where $g \circ f$ denotes the composition of the function $g$ to the output of the function $f$.

This channel still deletes each bit of the input i.i.d with probability $\delta$, but by defining the order of the deletion operations to begin with the last bit of the string, we ensure that we can easily track which portion of the channel was in charge of deleting the $i$th bit of the input message ($\Del{i}{\Bernoulli(\delta)}$).
This will help us avoid annoying technicalities later in the proof.

Similarly, we define the operation $\Insert{i}{j^\prime}$ to be the operation that inserts $j^\prime$ (i.i.d. uniformly distributed) bits before the $i$th bit of the input message.
We use this definition of the insertion operation to define the insertion channel with the same indexing convention:
\[
\InsChannel_{\eta} = \Insert{1}{\Geom(1 - \eta)} \circ \Insert{2}{\Geom(1 - \eta)} \circ \cdots
\]

An important property of these single bit insertion and deletion operations that will come in handy later is that their composition is ``additive'' in the sense that deleting $j_1$ bits and then an additional $j_2$ bits is the same as deleting $j_1 + j_2$ bits. 
In other words, 
\begin{equation}
    \label{eq:insdel_are_additive}
    \begin{aligned}
        &\Del{i}{j_1} \circ \Del{i}{j_2} = \Del{i}{j_1 + j_2}\\
        &\Insert{i}{j^\prime_1} \circ \Insert{i}{j^\prime_2} = \Insert{i}{j^\prime_1 + j^\prime_2}
    \end{aligned}
\end{equation}

\subsection{QP Decomposition of the Insertion and Deletion Channels}
\label{subsubsec:QP_decomp_BDC}

We now turn to the proof of Lemma~\ref{lem:qp_single_insdel}.

\paragraph{Proof Strategy}

Over the next few paragraphs, we will construct a QP-decomposition for the inverse of the individual deletion $\Lambda = \Del{i}{\Bernoulli(\delta)}$ and insertion $\Lambda^\prime = \Insert{i}{\Geom(1 - \eta)}$ channels.
Given access to the individual inverse channels, we will be able to invert the effects of either the insertion or the deletion channels on the first $k$ bits of the input one at a time and compute $f$ on the result.
For instance for the deletion channel we would have:

\begin{equation*}
    \begin{aligned}
        f_{\QP} \left(\wt{\bfx}\right) &= f\left(\left[\left(\Del{k}{\Bernoulli(\delta)}\right)^{-1} \circ \cdots \circ \left(\Del{1}{\Bernoulli(\delta)}\right)^{-1}\right]\left(\wt{\bfx}\right)\right) \\
        &= f\left(\left[\left(\Del{k}{\Bernoulli(\delta)}\right)^{-1} \circ \cdots \circ \left(\Del{1}{\Bernoulli(\delta)}\right)^{-1} \circ \textup{Deletion Channel}_q\right]\left(\bfx\right)\right) \\
        &= f\left(\left[\left(\Del{k}{\Bernoulli(\delta)}\right)^{-1} \circ \cdots \circ \left(\Del{1}{\Bernoulli(\delta)}\right)^{-1} \circ \Del{1}{\Bernoulli(\delta)} \circ \Del{2}{\Bernoulli(\delta)} \circ \cdots\right]\left(\bfx\right)\right) \\
        &= f\left(\left[\Del{k+1}{\Bernoulli(\delta)} \circ \cdots\right]\left(\bfx\right)\right) = f(\bfx)
    \end{aligned}
\end{equation*}

One technical difficulty we need to overcome in our proof is that due to the asynchronous nature of the deletion channel, there is a small probability that the QP oracle may try to sample more than $k^\prime$ bits of its input.
We ensure that this happens with a negligible probability (see Claim~\ref{clm:no_overflows}), and in these rare case, we return $0$ to avoid an overflow, potentially causing a small bias $\varepsilon$ instead.

\paragraph{The QP Set}
Recall from Section~\ref{sec:PEC} that the first step in a QP decomposition is to select a QP set -- a set of operations that we will use to span the inverse channel.
Moreover, as we noted in Section~\ref{sec:PEC}, a typical QP set will be comprised of operations similar to those already present in the noise channel.

More concretely, our QP set for the deletion channel will be the individual deletion channels $\Lambda_{j} = \Del{i}{j}$ and for the insertion channel our QP set will be the individual insertion channels $\Lambda^\prime_{j^\prime} = \Insert{i}{j^\prime}$.

\paragraph{The QP Decomposition}
Our QP set uses the same set of errors as our original noise channel.
The original noise channel $\Lambda = \Del{i}{\Bernoulli(\delta)}$ had some probability $p_j = 1 - \delta, \delta, 0, \ldots$ of applying each individual number of deletions $\Del{i}{0}, \Del{i}{1}, \Del{i}{2}, \ldots$.

Similarly, the QP decomposition assigns some quasi-probabilities $q_0, q_1,\ldots$ to the channels in the QP set $\Lambda_{0} = \Del{i}{0}, \Lambda_{0} = \Del{i}{1}, \ldots$.
Recall from Section~\ref{sec:PEC} that our goal in a QP decomposition is to ensure that:
\[
\sum_{j} q_{j} \Lambda_{j} = \Lambda^{-1}
\]
Or alternatively:
\begin{equation}
    \label{eq:adding_PEC_deletions}
    \begin{aligned}
        I = \Lambda_{0} = \sum_{j} q_{j} \Lambda_{j} \Lambda &= \sum_{i, j} p_{i} q_{j}  \Lambda_{j} \Lambda_{i} = \sum_{i, j} p_i q_{j} \Lambda_{i + j}
    \end{aligned}
\end{equation}
Equation~\eqref{eq:adding_PEC_deletions} gives us the necessary and sufficient conditions for this QP-decomposition to work:
\begin{equation}
    \label{eq:QP_condition_ugly}
    \forall \ell\in\N \;\;\;\; c_\ell \defeq \sum_{i + j = \ell} q_{j} p_{i} = \begin{cases}
        1 & \ell = 0\\ 
        0 & \text{otherwise}
    \end{cases}
\end{equation}

From here we could declare that $q_j$ are selected to be the solution to the set of linear equations~\eqref{eq:QP_condition_ugly}.
However, this would give us very little intuition for what these coefficients mean and how the QP norm might depend on the parameters of the channel.

Using the QP framework, we would say that the total number of deletions $\ell$ is equal to the sum of those created by the quasi-probabilistic process $j$ and those created by the deletion channel itself $i$.
Moreover, these are independently distributed.
Therefore, the quasi-probability distribution $c_\ell$ is simply the convolution of the quasi-probability distribution $q_j$ of the number of deletions made by the QP method and the probabilities $p_i$ of the deletion channel.

\[
c = q \star p
\]

Let $\mcG$ denote a generating function.
Extending the definition of generating functions to quasi-probabilities in the natural way ($\mcG_{\QP} (z) = \sum_j q_j z^j$), it is easy to show that we can still use the common identity on the generating function of the sum of independent variables to get equation~\eqref{eq:QP_decomposition_BDC_1}.

\begin{equation}
    \label{eq:QP_decomposition_BDC_1}
    1 = \mcG_{c} (z) = \mcG_{q} (z) \cdot \mcG_{\Bernoulli(\delta)} (z) \Rightarrow \mcG_{q} (z) = \left(\mcG_{\Bernoulli(\delta)} (z)\right)^{-1} = \frac{1}{(1-q) + qz}
\end{equation}

Equation~\eqref{eq:QP_decomposition_BDC_1} can be used to derive a QP decomposition of the deletion channel from the formula for the sum of a geometric series.

\begin{equation}
    \label{eq:QP_decomposition_BDC_2}
    \mcG_{q} (z) = \frac{1}{(1-\delta) + \delta z} = \frac{1}{1 - \delta}  \sum_j \left( \frac{-\delta}{1 - \delta} \right)^j z^j \Rightarrow \forall j \;\;\;\; q_j = \frac{1}{1 - \delta}\left( \frac{-\delta}{1 - \delta} \right)^j
\end{equation}

Similarly, for the insertion channel, we decompose the inverse of the insertion channel $\Lambda^\prime = \Insert{i}{\Geom(1 - \eta)}$ into a quasi-probability distribution $c^\prime_j$ over the insertion operations $\Lambda^\prime_j = \Insert{i}{j}$.
A similar analysis would yield the following values for the QP-decomposition of the insertion channel:

\begin{equation}
    \label{eq:QP_decomposition_Insertion}
    \mcG_{q^\prime} (z) \defeq \sum_j q_j^\prime z^j = \frac{1}{\mcG_{\Geom(1 - \eta)} (x)} = \frac{1 - \eta x}{1 - \eta} \Rightarrow q^\prime = \Bernoulli\left( \frac{-\eta}{1 - \eta} \right)
\end{equation}

Finally, we compute the QP norms of these decompositions:

\begin{equation}
    \label{eq:QP_norms_InsDel}
    \begin{aligned}
        &\gamma_{\text{Deletion}} = \frac{1}{1 - \delta} \sum_j \abs{\left( \frac{-\delta}{1 - \delta} \right)^j} = \frac{1}{1 - 2\delta}\\
        &\gamma_{\text{Insertion}} = \frac{1 + \eta}{1 - \eta}
    \end{aligned}
\end{equation}

The most important property of these QP norms for our usecase is that for any $\delta < 1/2$ and $\eta \in [0, 1)$ these norms are finite and constant.
This property allows us to use this QP decomposition with a finite overhead.

\paragraph{The QP Oracle}

From here we can use the Monte-Carlo sampling described in Section~\ref{sec:PEC} (see Claim~\ref{clm:sample_qp}).
By applying it to the QP decomposition of the inverse deletion / insertion channel on each of the first $k$ bits of the trace, we can invert the individual components affecting the first $k$ bits of the input message one after another and apply $f$ to the result.

This method samples each set of selections for the number of deletions or insertions $\left(j_1, \ldots, j_k\right)$ with probability $p_{\vec{j}} \propto q_{\vec{j}} \defeq \prod_i q_{j_i}$, applies $f$ to the output and multiplies the result by the scaling factor $q_{\vec{j}} / p_{\vec{j}} = \sgn{q_{\vec{j}}} \gamma$.
Because we are only inverting the last $k$ steps of the channel (we assumed the channel went over the input string in reverse, so these are the components of the channel affecting the first $k$ bits of the message), we have
\[
\gamma=\begin{cases}
    \gamma_{\text{Deletion}}^k = \left(\frac{1}{1 - 2\delta}\right)^k \\
    \gamma_{\text{Insertion}}^k = \left(\frac{1+\eta}{1 - \eta}\right)^k
\end{cases}
\]
In both cases $\gamma = \exp(O(k))$, completing the proof of Lemma~\ref{lem:qp_single_insdel}.

\paragraph{Avoiding Buffer Overflows}

A final technical detail that we need to consider is the fact that our QP decomposition of the deletion channel may also perform additional deletions.
This could cause our QP oracle to access bits beyond the $k$th bit of its input string.

To avoid this, we modify the QP oracle in two ways.
First, we allow it to access $k^\prime = k + 4 \cdot \frac{\delta}{1-2\delta} k + 4 \cdot \frac{\delta}{1 - \delta} \Paren{\log(\gamma) + \log\left(\varepsilon^{-1}\right)} > k$ bits of the input.
This will ensure that the probability of a ``buffer overflow'' (i.e., of us reading beyond the input length of $f_{\QP(\mcC)}$ is very small).
Secondly, we add a test to our QP oracle to avoid this very rare edge case (and have it return the arbitrary value of $0$).

\begin{claim}
\label{clm:no_overflows}
The probability that the QP sampling process above goes beyond the $k^\prime$th bit of the input is at most $\frac{\varepsilon}{\gamma}$.
\end{claim}

Claim~\ref{clm:no_overflows}, promises that we trigger the buffer-overflow condition in Algorithm~\ref{alg:qp_single_insdel} with a probability of at most $\frac{\varepsilon}{\gamma}$.
Therefore, replacing the output of the QP estimator (which is otherwise bounded by $\gamma C$) with $0$ on these cases will shift its expectation by an additive error of at most $C \varepsilon$.
Claim~\ref{clm:no_overflows} follows almost immediately from the following concentration bound on the sum of i.i.d. geometric random variables:

\begin{lemma}
\label{lem:bound_on_sum_geometric}
Let $X_1, \ldots, X_k \sim \text{Geom}(p)$ be independent and identically distributed geometric random variables for some $p \in (0, 1)$. Then, for any $\varepsilon > 0$,
\[
\Prob{}{\sum_{j=1}^k X_j > 4 \frac{1-p}{p} k + \Delta} < e^{-c \Delta}\,,
\]
where $c = \min \set{1, \frac{p}{4(1-p)}}$.
\end{lemma}

\begin{proof}[Proof of Claim~\ref{clm:no_overflows} using Lemma~\ref{lem:bound_on_sum_geometric}]

In each step of the QP sampling process in Algorithm~\ref{alg:qp_single_insdel}, we delete $\Geom\Paren{1 - \frac{\delta}{1 - \delta}}$ bits of the input and then output an additional $1$ bit.

Applying Lemma~\ref{lem:bound_on_sum_geometric} with $p = 1 - \frac{\delta}{1 - \delta}$, proves that after $k$ iterations of this QP channel, the total number of deletions is at most $k^\prime - k$ with probability $1 - \frac{\varepsilon}{\gamma}$.

\end{proof}

\begin{proof}[Proof of Lemma~\ref{lem:bound_on_sum_geometric}]
First, recall that the moment-generating function (MGF) of a geometric random variable $X_i \sim \text{Geom}(p)$ is given by:
\[
M_{X_i}(t) = \EE{e^{t X_i}} = \frac{p}{1 - (1 - p) e^t}, \quad \text{for } t < -\ln(1 - p).
\]

\textbf{Applying the Chernoff Bound:}

From the Markov inequality on $Z = \exp\Paren{t \sum_{i = 1}^k X_i}$, we have that
\[
\Prob{}{\sum_{j=1}^k X_j > C} < \frac{\EE{Z}}{e^{tC}} = e^{- t C} \left(M_{X_i}(t)\right)^k = e^{- t C}\Paren{\frac{p}{1 - (1-p)e^t}}^k\,.
\]

For any $t < -\ln(1 - p)$, the Chernoff bound states:
\[
\Prob{}{\sum_{j=1}^k X_j > C} \leq e^{-t C} \left(M_{X_i}(t)\right)^k = e^{-t C} \left(\frac{p}{1 - (1 - p) e^t}\right)^k.
\]

\textbf{Calculating the Bound:}

The next step of finding a Chernoff bound is to pick a good value of $t$.
In our case, we will set $t = \min \set{1, \frac{p}{4(1-p)}}$.
When $0 < p \leq 3/4$, we have
\[
t \leq p < -\ln(1-p)\,,
\]
and when $1 > p > 3/4$, we have
\[
t \leq 1 < -\ln(1-p)\,.
\]

Because $t \leq 1$, we have $e^t < 1 + 2 t$, so
\[
\frac{p e^t}{1 - (1 - p) e^t} < \frac{p}{1 - (1-p)(1+2t)} = \frac{1}{1 - 2\frac{1-p}{p}t}\,.
\]

Because $t \leq \frac{p}{4(1-p)}$, we have $2\frac{1-p}{p}t \leq \frac{1}{2}$, so
\[
 \frac{1}{1 - 2\frac{1-p}{p}t} < 1 + 4 \frac{1-p}{p}t < \exp \Paren{4 \frac{1-p}{p}t}\,.
\]

Therefore
\[
\Prob{}{\sum_{j=1}^k X_j > C} \leq e^{-t C} \left(\frac{p}{1 - (1 - p) e^t}\right)^k < \exp\Paren{4 \frac{1-p}{p} t - tC} = \exp\Paren{-t \Delta}\,.
\]

\end{proof}

\subsection{Extracting a \texorpdfstring{$k$}{k}-mer Estimation Algorithm}
\label{subsec:generating_funcs}

The final component of our analysis will be to show that applying Lemma~\ref{lem:qp_single_insdel} as in Algorithm~\ref{alg:qp_kmer_estimation} does indeed yield a good $k$-mer estimation algorithm.

\begin{claim}
    \label{clm:alg2good}
    Algorithm~\ref{alg:qp_kmer_estimation} is a good $k$-mer estimation algorithm.
\end{claim}

\begin{proof}[Proof of Claim~\ref{clm:alg2good}]

Recall that in our setting, we have traces $\wt{\bfx}$ drawn independently from an application of the noise channel $\mcC = \mcC_1 \circ \cdots \circ \mcC_L$ to an input string $\bfx$.
This channel maps each bit of the input string into a distribution of lengths $\abs{\mcC}$ in the output string.
Let $j_l \sim \abs{\mcC}$ denote the number of bits of the output string originating from the $j$th bit of the input.

Let $J_i = j_1 + \cdots + j_i$ denote their cummulative sums, and consider the event $E_{l \rightarrow j} = \set{J_{l-1} \leq j} \cap \set{J_{l} \geq j}$ denote the event that the $j$th bit of the output string was part of the output of the channel on the $l$th bit of the input.
Conditioned on this event, the $j$-suffix of the trace is distributed according to a trace of the $l$-suffix of the input string
\begin{equation}
\label{eq:trace_suffix}
    \left( \Paren{\wt{\bfx}}_{j:} \mid E_{l \rightarrow j} \right) \sim \wt{\Paren{\bfx_{l:}}}
\end{equation}

Next, we will use equation~\eqref{eq:trace_suffix} to relate the output of our QP estimator to the $k$-mer estimation value, using an analysis similar to the one used quite often in trace reconstruction literature~\cite{nazarov2017trace,de2017optimal,holden2020subpolynomial,rubinstein2022average}.
Recall that the goal of our $k$-mer estimation is to approximate the following value:
\[
k\text{-mer Value}_{\zeta, \bfw} \defeq \sum_{l \in [n]} \zeta^{l} 1_{\bfx_{l:l+k} = \bfw}
\]

Let $\mu_j$ denote the empirical average of the QP estimator on the $j$-suffixes of the traces.
As we have shown that our QP estimator is a good estimator, we have
\[
\mu_j = \Expect{\wt{\bfx} \sim \mcC(\bfx)}{f_{\QP(\mcC)}\left( \wt{\bfx}_{j:} \right)} \pm \mcE = \sum_{l \in [n]} \Prob{\wt{\bfx} \sim \mcC(\bfx)}{ E_{l \rightarrow j}} 1_{\bfx_{l:l+k} = \bfw} \pm (\varepsilon + \mcE)
\]
where $\mcE$ is the error due to the empirical mean estimation and $\varepsilon$ is the error allowance of the QP estimator.

Let $G_{\abs{\mcC}}(\zeta)$ denote the generating function of the distribution $\abs{\mcC}$.
From basic identities on generating functions, the generating function of a composition of multiple channels is the composition of the generating functions
\[
G_{\abs{\mcC}} = G_{\abs{\mcC_L}} \circ \cdots \circ G_{\abs{\mcC_1}}\; .
\]
This is because each bit in the output of the first channel $\mcC_L$ is then transformed into $\abs{\mcC_{L-1}}$ bits by the second channel, etc.

Because the channel is applied to each of the bits independently, we have
\begin{equation*}
    \begin{aligned}
        \sum_{j \in \N} \mu_j z^j &\approx \sum_{j \in \N} z^j \sum_{l \in [n]} \Prob{}{E_{l \rightarrow j}} 1_{\bfx_{l:l+k} = \bfw} =\\
        &= \sum_{j \in \N} \sum_{j^\prime \in \N} z^{j+j^\prime} \sum_{l \in [n]} \Prob{}{J_{l-1} = j} \Prob{}{j_l = j^\prime} 1_{\bfx_{l:l+k} = \bfw} =\\
        &= G_{\abs{\mcC}}(z) \sum_{l \in [n]} G_{\abs{\mcC}}(z)^{l-1} 1_{\bfx_{l:l+k} = \bfw}
    \end{aligned}
\end{equation*}
(where the $\approx$ suppresses the $O(\mcE+\varepsilon)$ terms due to the bias of our QP estimator and mean estimation error).
Therefore, setting $z = G^{-1}_{\abs{\mcC}}(\zeta)$, we have
\[
E = \zeta^{-1} \sum_{j \in \N} \mu_j z^j \approx k\text{-mer Value}
\]

In order to complete the proof, it remains to show that the errors and biases of this estimation are small even with a moderate sample complexity.
Clearly with very high probability, none of the traces will be of length greater than some $O(n)$.
Therefore, the errors in this approximation scale like $z^{O(n)} (\mcE + \varepsilon)$, while a combination of the bounds on the image of $\abs{f_{\QP(\mcC)}}$ in Lemma~\ref{lem:qp_single_insdel} and Hoeffding give us that with probability $\geq 1 - \delta$
\begin{equation}
\label{eq:empirical_mean_error}
    \abs{\mcE} \leq \frac{\exp(O(k)) \times \poly(\varepsilon^{-1}) \times \poly(\log(\delta^{-1}))}{\sqrt{\NumSamples}}
\end{equation}

Therefore, all that remains is to show that $z$ can be computed efficiently and is not too large in absolute value.
Note that for our choices of error channels, each of the individual generating functions $G_{\abs{\mcC_i}}$ is a M\"{o}bius transformation, and therefore, so is their composition.
Because $G_{\abs{\mcC}}$ is a M\"{o}bius transformation, we may easily invert it, and as it is a generating function, $1$ is a fixed point of $G_{\abs{\mcC}}$ and it has a real (i.e., not complex) gradient at that point.
Therefore, from a simple Taylor series argument ($G^{-1}_{\abs{\mcC}}$ is also a M\"{o}bius and is analytical near $1$), we have
\[
G_{\abs{\mcC}}^{-1}(e^{i \alpha}) = 1 + O(\alpha) i + O(\alpha^2)
\]
In other words, moving along the unit circle in the neighborhood of $1$, the absolute value of $G_{\abs{\mcC}}^{-1}$ grows quadratically with $\alpha$.

Therefore, for $z = G_{\abs{\mcC}}^{-1}(\zeta)$, we have $\abs{z} = (1 + O(\alpha^2))$, yielding a good bound on the sample complexity of our $k$-mer estimation algorithm.
    
\end{proof}

\section{Efficient Trace Reconstruction Algorithm}
\label{sec:linear_programming}
Over the last few sections of this paper, we have constructed a $k$-mer estimation algorithm for a general class of trace reconstruction problems.
In this section, we will show that this $k$-mer estimation implies a good trace reconstruction algorithm, proving Theorem~\ref{thm:main_res_inf}.

\begin{proof}[Proof of Theorem~\ref{thm:main_res_inf}]
    We will adapt the linear programming approach of Holenstein et al.~\cite{holenstein2008trace} from the mean-based setting to the $k$-mer setting.
Consider the variables $v_{\bfw, l} = 1_{\bfx_{l:l+k} = \bfw} \in \set{0, 1}$.
Given access to the value of $v_{\bfw, l}$ for all $\bfw \in \bits^k$ and $l \in [n]$, we can clearly reconstruct $\bfx$.
We will attempt to find these values by solving a relaxed version of this search with variables $v_{\bfw, l} \in [0, 1]$.
It will suffice to reconstruct $v_{\bfw, l}$ for all aperiodic $\bfw$, because for any prefix $\bfx_{:l}$, at least one of the completions 
\[
\bfx_{:l+1} = \begin{cases}
    \bfx_{:l} \mid 0\\
    \bfx_{:l} \mid 1
\end{cases}
\]
must end with an aperiodic suffix (see Lemma~\ref{lem:sparse_substring}).
This aperiodicity implies a sparsity condition giving us the inequality
\[
\forall l \in [n-k / 2], \bfw \in \bits^k\;\;\;\; \sum_{j \in [l, l + k/2]} v_{\bfw, j} \leq 1
\]

Moreover, each frequency $\alpha \in [\pm \wt{O}(n^{1/5})]$, the $k$-mer estimation on $\bfw$ and $\alpha$ on the $v_{\bfw , l}$.

We will show that we can reconstruct the values of the aperiodic $v_{\bfw, l}$ one at a time (i.e., we fix some aperiodic $\bfw \in \bits^k$ and reconstruct the values of $v_{\bfw, l}$ for $l = 1, 2, \ldots$), by solving a small number of linear programs using these linear constraints (i.e., the sparsity constraint and the $k$-mer estimation constraint) for each $l$.

We do this by induction over $l$.
Suppose that either $l = 1$, or that we have already recovered the first $l-1$ entries of $v_{\bfw, l}$.

Chase~\cite{chase2021separating} proved that
\begin{theorem} [Theorem 5 of \cite{chase2021separating}]
\label{thm:like_borwein}
    Let $\mathcal{P}_n^\mu$ denote the set of polynomials of the form $p(x) = 1 - \eta x^d + \sum_{n^\mu \leq j \leq n} a_j x^j$ where $\eta \in \{0, 1\}$ and $\abs{a_j} \leq 1$.
    
    For any $\mu \in (0, 1)$, there exists some constant $C_1 > 0$, such that for all sufficiently large $n$, any $p \in \mathcal{P}_n ^ \mu$:
    \[
    \max_{\lvert \alpha \rvert \leq n^{-2\mu}} \lvert p(e^{i\alpha}) \rvert \geq \exp\left(-C_1 n^\mu \log^5 (n)\right)
    \]
\end{theorem}

To use Theorem~\ref{thm:like_borwein}, suppose we consider the two hypotheses for the next entry of $v$.
In particular, let $v_{\bfw, :l-1} = v^\prime_{\bfw, :l-1}$, and suppose we add the constraints $v_{\bfw, l} = 1$ and $v^\prime_{\bfw, l + t} = 1$ for some specific $t \leq k / 2$ or the constraints $v^\prime_{\bfw, l + t} = 0$ for all $t < k/2$.

Theorem~\ref{thm:like_borwein} implies that for any set of values we assign to the other entries of $v_{\bfw, l+1:}$ and $v^\prime_{\bfw, l+1:}$, there exists some frequency $\alpha \in [-n^{-2\mu}, n^{-2\mu}]$ for which
\begin{equation}
    \label{eq:good_alpha}
    \abs{\sum_{j \in [n]} v_{\bfw, j} e^{i j \alpha} - \sum_{j \in [n]} v^\prime_{\bfw, j} e^{i j \alpha}} \geq \exp\left(-C_1 n^\mu \log^5 (n)\right)
\end{equation}

The order of the quantifiers in the statement above will make it difficult to proceed (we don't want to have $\alpha$ dependent on the valuations of $v, v^\prime$), but it is easy to see that taking a fine-net of $\alpha_m = -n^{-2\mu} + \frac{m}{10 \exp\left(-C_1 n^\mu \log^5 (n)\right) n^2}$ suffices.
More concretely, note that
\[
\frac{\partial}{\partial \alpha}  \left( \sum_{j \in [n]} v_{\bfw, j} e^{i j \alpha} - \sum_{j \in [n]} v^\prime_{\bfw, j} e^{i j \alpha} \right) < n^2
\]
so for the $\alpha_m$ closest to the value of $\alpha$ in equation~\eqref{eq:good_alpha}, we have
\[
\abs{\sum_{j \in [n]} v_{\bfw, j} e^{i j \alpha_m} - \sum_{j \in [n]} v^\prime_{\bfw, j} e^{i j \alpha_m}} \geq \frac{1}{2} \exp\left(-C_1 n^\mu \log^5 (n)\right)
\]

Therefore, at most one of $v$ or $v^\prime$ can have a solution that satisfies all the $k$-mer constraints, and we may proceed with this hypothesis of $v_{\bfw, l}$.
This process required at most $\exp(\wt{O}(n^{1/5})) \times \poly(n)$ linear programs, each requiring at most $\exp(\wt{O}(n^{1/5}))$ $k$-mer evaluations with $k = \wt{O}(n^{1/5})$, so the whole process had time and sample complexity at most $\exp(\wt{O}(n^{1/5}))$.
\end{proof}

We note that the algorithm above utilizes linear programming in a very similar way to the ones designed by Holenstein et al.~\cite{holenstein2008trace} for mean-based trace reconstruction.
While to the best of our knowledge, they have not yet been applied to $k$-mer based trace reconstruction, the main contribution here is that by making our $k$-mer estimation more robust, we are able to enjoy their lower time complexity even in the high-deletion regime, whereas it is not clear how to get a similar complexity from Chase's solution for the high-deletion regime.

\section{Discussion}\label{sec:discussion}
In this paper, we considered the $k$-mer estimation problem -- a key component in leading trace reconstruction results.
Previous algorithms for $k$-mer estimation were complicated and hard to adapt to minor changes in the noise model.

The ability to easily adapt our algorithm to more complex error models is crucial for any practical implementation of trace reconstruction algorithms, as it is highly unlikely that we will be faced with a pure form of the problem in practice.
While~\cite{cheraghchi2022mean} showed that we can adapt mean-based approaches to more general error models, it was not clear if $k$-mer based approaches could be similarly adapted.

Drawing inspiration from the field of quantum error mitigation, we introduce a simpler algorithm for $k$-mer estimation that can be easily applied to a wide variety of noise channels, allowing for time and sample efficient trace reconstruction from these channels.

Beyond the obvious question of improving our trace reconstruction algorithms to run with an even lower time / sample complexity, we posit the question of whether or not the quasi-probability method can be applied to additional learning theory problems.
The key properties of trace reconstruction and quantum error mitigation which made them amenable to this method are that we are given many samples from a known but random channel, that our goal is to estimate the mean of some property of the noise-free distribution, and that the main difficulty is in analysing a closed component of the error channel.

\section*{Acknowledgments}

Additionally, I would like to thank Zachary Chase for helpful comments on a previous version of this paper.

\bibliographystyle{alpha}
\bibliography{main}

\appendix

\end{document}